\documentclass[10pt,twocolumn]{article}

\usepackage[T1]{fontenc}
\usepackage{amsmath,amssymb,amsthm,mathtools}
\usepackage{newtxtext,newtxmath}
\usepackage[margin=0.72in]{geometry}
\usepackage{booktabs}
\usepackage{array}
\usepackage{graphicx}
\usepackage{tabularx}
\usepackage{cuted}
\usepackage{capt-of}
\usepackage{float}
\usepackage{xcolor}
\usepackage{microtype}
\usepackage{enumitem}
\usepackage{tikz}
\usetikzlibrary{arrows.meta,positioning,calc,fit}
\usepackage[numbers,sort,compress]{natbib}
\usepackage{xurl}
\usepackage[hidelinks]{hyperref}
\usepackage{orcidlink}

\graphicspath{{figures/}}
\newcommand{\method}{\mbox{\textsc{CavityRank}}}
\newcommand{\crtwo}{\mbox{\textsc{CR2}}}
\newcommand{\guard}{\mbox{\textsc{CavityRank-Path}}}
\newcommand{\indicator}[1]{\mathbf{1}\!\left\{#1\right\}}
\newcommand{\mset}[1]{\{\!\{#1\}\!\}}
\newcommand{\lane}[1]{\fbox{\makebox[1.12em][c]{\ensuremath{\mathstrut#1}}}}
\newcommand{\bucketfour}[4]{\lane{#1}\kern-\fboxrule\lane{#2}\hspace{0.55em}\lane{#3}\kern-\fboxrule\lane{#4}}
\definecolor{deepblue}{HTML}{1769AA}
\definecolor{deepgreen}{HTML}{00866B}
\definecolor{deepred}{HTML}{B23A2B}
\definecolor{midgray}{HTML}{666666}

\newtheorem{theorem}{Theorem}
\newtheorem{lemma}{Lemma}

\setlist[itemize]{leftmargin=*,itemsep=1.6pt,topsep=3pt,parsep=0pt}
\hypersetup{
  pdftitle={CavityRank: Zero-Extra-Byte Residual Routing for Cuckoo Filters},
  pdfauthor={Yongjie Guan}
}

\newcommand{\ResultCapHTwoIndependentGapClosurePct}{51.98\%}
\newcommand{\ResultCapHTwoIndependentGapItemsMean}{0.06}

\newcommand{\ResultCapHTwoXorSixteenGapClosurePct}{58.39\%}
\newcommand{\ResultCapHTwoXorSixteenGapItemsMean}{0.08}

\newcommand{\ResultCorrectnessRows}{192}

\newcommand{\ResultQueryEquivalentPairs}{16}

\newcommand{\CausalMainKeyedClosureCIHighPct}{91.41\%}
\newcommand{\CausalMainKeyedClosureCILowPct}{89.95\%}

\newcommand{\CausalMainKeyedClosurePct}{90.69\%}
\newcommand{\CausalMainKeyedGapMean}{4.02}
\newcommand{\CausalMainKeyedInteractionHighPP}{1.33}
\newcommand{\CausalMainKeyedInteractionLowPP}{-0.38}
\newcommand{\CausalMainKeyedInteractionPP}{0.47}
\newcommand{\CausalMainKeyedLosses}{23}

\newcommand{\CausalMainKeyedTies}{94}
\newcommand{\CausalMainKeyedWins}{1931}

\newcommand{\CausalMainXorCRFourGapSum}{914}
\newcommand{\CausalMainXorClosureCIHighPct}{90.67\%}
\newcommand{\CausalMainXorClosureCILowPct}{89.18\%}

\newcommand{\CausalMainXorClosurePct}{89.95\%}
\newcommand{\CausalMainXorGapMean}{3.99}
\newcommand{\CausalMainXorInteractionHighPP}{0.96}
\newcommand{\CausalMainXorInteractionLowPP}{-0.74}
\newcommand{\CausalMainXorInteractionPP}{0.11}
\newcommand{\CausalMainXorLosses}{15}
\newcommand{\CausalMainXorPairGapSum}{9092}

\newcommand{\CausalMainXorTies}{116}
\newcommand{\CausalMainXorWins}{1917}

\newcommand{\CausalScaleKeyedClosureCIHighPct}{84.94\%}
\newcommand{\CausalScaleKeyedClosureCILowPct}{83.55\%}

\newcommand{\CausalScaleKeyedClosurePct}{84.27\%}
\newcommand{\CausalScaleKeyedGapMean}{109.54}

\newcommand{\CausalScaleXorClosureCIHighPct}{84.77\%}
\newcommand{\CausalScaleXorClosureCILowPct}{83.13\%}

\newcommand{\CausalScaleXorClosurePct}{83.95\%}
\newcommand{\CausalScaleXorGapMean}{107.51}

\newcommand{\CausalMainXorReachDiffPP}{6.15}
\newcommand{\CausalMainXorReachDiffCILowPP}{5.13}
\newcommand{\CausalMainXorReachDiffCIHighPP}{7.23}
\newcommand{\CausalMainKeyedReachDiffPP}{6.84}
\newcommand{\CausalMainKeyedReachDiffCILowPP}{5.76}
\newcommand{\CausalMainKeyedReachDiffCIHighPP}{7.96}
\newcommand{\CausalScaleXorReachDiffPP}{41.80}
\newcommand{\CausalScaleXorReachDiffCILowPP}{35.94}
\newcommand{\CausalScaleXorReachDiffCIHighPP}{48.05}
\newcommand{\CausalScaleKeyedReachDiffPP}{44.92}
\newcommand{\CausalScaleKeyedReachDiffCILowPP}{38.67}
\newcommand{\CausalScaleKeyedReachDiffCIHighPP}{51.17}

\title{CavityRank: Zero-Extra-Byte Residual Routing for Cuckoo Filters}
\author{Yongjie Guan\,\orcidlink{0009-0002-1071-2741}\\
  \small Zhejiang University of Technology}
\date{}

\begin{document}
\linespread{0.978}\selectfont
\maketitle
\raggedbottom

\begin{abstract}
Near capacity, a cuckoo filter may reject an insertion even though a legal
placement still exists: the table remains structurally feasible, but a bounded
policy fails to find an augmenting path. Random kick-out keeps each step cheap
but leaves no persistent direction; breadth-first search recovers direction by
expanding a frontier and maintaining table-scaled state. \method{} exploits a
resource already present in four-slot packed buckets. Lookup observes only the
fingerprint multiset, so query-equivalent lane orders can encode two comparison
bits without widening the 64-bit bucket or changing the two-bucket query. The
bits form a four-level ordinal residual rank. Insertion follows a minimum-rank
edge and re-encodes each modified bucket from its outgoing edges after
relocation, propagating the rank actually realized by the packed word. An exact
capacity-four orientation oracle separates structural infeasibility from
bounded-search loss. In a paired 4,096-bucket XOR16 ladder, \crtwo{} closes
86.47\% of Random CF's oracle gap and \method{} leaves 1.39\% of that original
gap. A canonical-tie \crtwo{}-versus-\method{} ablation isolates the second
implicit bit, which closes \CausalMainXorClosurePct{} and
\CausalMainKeyedClosurePct{} of \crtwo{}'s residual gap; the corresponding
closures at 65,536 buckets are
\CausalScaleXorClosurePct{} and \CausalScaleKeyedClosurePct{}. A separate packed
implementation study at 64~MiB and 97.75\% load records 42.53 logical reads per
insertion, versus 62.67 for explicit labels and 355.74 for depth-10 BFS, with
zero extra bytes per bucket and no table-scaled workspace. \method{} therefore
occupies a practical design point between unguided eviction and frontier
search.
\end{abstract}

\begin{strip}
  \centering
  \begin{minipage}{\textwidth}
  \centering
\begin{tikzpicture}[
  font=\small,
  bucket/.style={circle,draw,minimum size=6.5mm,inner sep=0pt},
  cavity/.style={circle,draw,very thick,minimum size=8mm,inner sep=0pt},
  edge/.style={-{Latex[length=2.2mm]},gray!55,thin},
  hot/.style={-{Latex[length=2.5mm]},deepred,very thick},
  route/.style={-{Latex[length=2.5mm]},deepblue,very thick}
]
\begin{scope}[xshift=0cm]
  \node[font=\bfseries] at (2.55,2.65) {Unguided path};
  \node[bucket] (a) at (0.55,1.45) {};
  \node[bucket] (b) at (1.9,2.05) {};
  \node[bucket] (c) at (1.9,0.65) {};
  \node[bucket] (d) at (3.55,1.45) {};
  \node[cavity] (e) at (4.8,0.65) {\scriptsize cavity};
  \draw[edge] (a)--(b); \draw[edge] (a)--(c);
  \draw[edge] (b)--(d); \draw[edge] (c)--(d); \draw[edge] (d)--(e);
  \draw[hot] (a)--(b);
  \draw[hot,bend left=18] (b) to (d);
  \draw[hot,bend left=18] (d) to (b);
  \node[text=deepred,font=\Large\bfseries] at (3.5,2.2) {$\times$};
\end{scope}
\begin{scope}[xshift=6.2cm]
  \node[font=\bfseries] at (2.55,2.65) {CavityRank};
  \node[bucket] (a2) at (0.35,1.35) {4};
  \node[bucket] (b2) at (1.55,2.05) {4};
  \node[bucket] (c2) at (1.55,0.65) {3};
  \node[bucket] (d2) at (2.85,1.35) {2};
  \node[bucket] (f2) at (3.95,0.95) {1};
  \node[cavity] (e2) at (5.05,0.65) {\scriptsize cavity};
  \draw[edge] (a2)--(b2); \draw[edge] (a2)--(c2);
  \draw[route] (a2)--(c2); \draw[route] (c2)--(d2);
  \draw[route] (d2)--(f2); \draw[route] (f2)--(e2);
\end{scope}
\begin{scope}[xshift=12.4cm]
  \node[font=\bfseries] at (2.55,2.65) {Order stores rank};
  \node[anchor=west] at (0.65,2.05)
    {{\bfseries 1}\quad\bucketfour{a}{b}{c}{d}};
  \node[anchor=west] at (0.65,1.50)
    {{\bfseries 2}\quad\bucketfour{b}{a}{c}{d}};
  \node[anchor=west] at (0.65,0.95)
    {{\bfseries 3}\quad\bucketfour{a}{b}{d}{c}};
  \node[anchor=west] at (0.65,0.40)
    {{\bfseries 4}\quad\bucketfour{b}{a}{d}{c}};
\end{scope}
\end{tikzpicture}

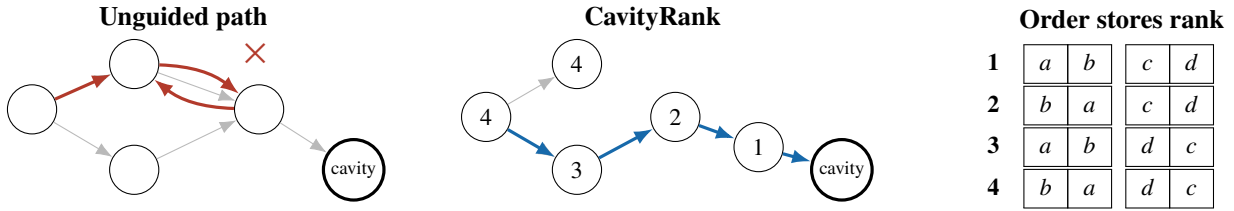
\captionof{figure}{Bucket order supplies persistent residual direction without
changing the resident multiset or query. An unguided path can revisit full
buckets and miss an existing route to a cavity; \method{} instead follows
minimum-ranked alternates. For $a<b<c<d$, two lane-pair orientations encode
ranks 1--4 while storing the same four fingerprints.}
\label{fig:teaser}
  \end{minipage}
\end{strip}

\section{Introduction}
\label{sec:intro}

A cuckoo filter can reject an insertion while the current item prefix still
admits a legal placement. At low load this distinction is rarely visible. Near
the capacity boundary, however, one bounded failure can trigger a rebuild or
expansion even though the table still has structural capacity. The practical
problem is therefore not only how much a filter can hold, but how much of that
capacity an online insertion policy can reach without widening the bucket or
complicating lookup.

Existing policies occupy two useful extremes. Random or rotating eviction
follows one inexpensive relocation path, but the table retains little
information about where a free slot lies. Breadth-first relocation finds short
augmenting paths by expanding a frontier \citep{li2014algorithmic}; a practical
implementation also maintains visited and parent state proportional to the
table. The missing design point is a guided single path: persistent enough to
retain residual direction, yet compact enough to preserve the packed filter
representation.

The key observation is that a full bucket has two views. Membership lookup
observes the multiset of four fingerprints; insertion can also observe their
order. Lane permutations preserve the multiset, so query-equivalent orders form
a latent state space inside every full bucket. \method{} turns two pair
orientations into two implicit bits and a four-level ordinal residual rank. The
rank is not stored beside the payload. It is the payload order.

\method{} remains a one-path relocation policy. At each full bucket, it reads
the ranks of the four alternate targets and follows a minimum-rank edge. After
a swap, it recomputes the current bucket from the residual graph that actually
exists: the victim edge is gone, the incoming fingerprint contributes an edge
back to its predecessor, and the other three edges remain. The encoder then
propagates the rank realized by the packed word, including duplicate-fingerprint
cases in which a requested state aliases to another state.

Evaluation presents a second challenge. Maximum achieved load conflates two
events: the prefix may be structurally infeasible, or bounded search may have
missed a legal orientation. We pair every online run with an exact
capacity-four orientation oracle and measure the resulting online-to-oracle
gap. A paired 4,096-bucket ladder shows that \crtwo{} first closes 86.47\% of
Random CF's XOR16 gap and \method{} leaves 1.39\% of the original gap. A
controlled \crtwo{}-versus-\method{} comparison then changes only the implicit
alphabet from one bit to two. At 4,096 buckets, the second bit closes
\CausalMainXorClosurePct{} of the XOR16 gap and
\CausalMainKeyedClosurePct{} of the Keyed-XOR16 gap. At 65,536 buckets, it still
closes \CausalScaleXorClosurePct{} and \CausalScaleKeyedClosurePct{}. A separate
64~MiB packed implementation uses 42.53 logical reads per insertion at 97.75\%
load, 32\% fewer than explicit labels and 88\% fewer than depth-10 BFS, with no
extra per-bucket bytes or table-scaled workspace.

This paper makes three contributions:
\begin{itemize}
  \item \textbf{Implicit residual-rank encoding.} \method{} encodes a
  four-level routing rank in the query-equivalent lane order of a standard
  64-bit four-slot bucket, requiring neither wider buckets nor auxiliary
  metadata.
  \item \textbf{Residual routing with post-state semantics.} A minimum-rank
  single path and a unified post-relocation backup maintain a bounded ordinal
  signal while updating control state only in buckets modified by insertion.
  \item \textbf{Exact separation of capacity and search.} A complete
  orientation oracle isolates bounded-search loss from structural
  infeasibility. An end-to-end policy ladder locates the gain; controlled
  ablations, scale confirmation, and packed measurements isolate the second
  bit and quantify its machine cost.
\end{itemize}

\section{Capacity, Search, and Bucket Order}
\label{sec:model}

\subsection{Packed Partial-Key Placement}

Cuckoo hashing assigns each item a small set of candidate locations and resolves
collisions by relocating residents
\citep{pagh2004cuckoo,dietzfelbinger2007balanced}. A cuckoo filter stores only a
fingerprint $f$ while retaining alternate recovery and two-bucket lookup
\citep{fan2014cuckoo}. Eppstein later analyzed partial-key cuckoo filters
through a graph whose vertices are buckets and whose fingerprint edges join
their two candidate locations \citep{eppstein2016cuckoo}; assigning each edge
to its resident endpoint gives the orientation view used below. We study the
common four-slot packed representation. The table has a power-of-two number
$m$ of buckets; each bucket contains up to four nonzero \texttt{u16}
fingerprints in one \texttt{u64}; zero denotes an empty lane.

Hashing key $x$ yields a primary bucket $i_1(x)$ and fingerprint $f(x)$. Its
second bucket is
\begin{equation}
  i_2(x)=i_1(x)\mathbin{\oplus}\Delta(f(x)),
  \label{eq:alternate}
\end{equation}
where the masked offset $\Delta(f)$ lies in $\{1,\ldots,m-1\}$. Writing
$\operatorname{alt}(i,f)=i\mathbin{\oplus}\Delta(f)$, XOR makes alternate
recovery an involution:
\begin{equation}
  \operatorname{alt}(\operatorname{alt}(i,f),f)=i.
  \label{eq:alternate-involution}
\end{equation}
A resident fingerprint and its current bucket therefore reveal both legal
endpoints without retaining the original key.

\begin{theorem}[Query-equivalent bucket order]
\label{thm:query-invariance}
Permuting the four nonzero lanes of any full bucket preserves placement
validity, occupancy, payload footprint, and membership answers.
\end{theorem}
\begin{proof}
A permutation changes neither the resident multiset nor the bucket containing
any fingerprint, so every resident remains at a legal endpoint and the packed
word remains 64 bits. If $a_0,\ldots,a_3$ are the fingerprints in one candidate
bucket, $f$ is the query fingerprint, and $\pi$ is a permutation of
$\{0,1,2,3\}$, then
\begin{equation}
  \sum_{j=0}^{3}\indicator{a_j=f}
  =
  \sum_{j=0}^{3}\indicator{a_{\pi(j)}=f}.
  \label{eq:query-invariance}
\end{equation}
Thus the number of matching lanes, and hence the bucket-local membership
answer, is unchanged. The identity holds independently in both candidate
buckets.
\end{proof}

Theorem~\ref{thm:query-invariance} creates the design space used by
\method{}: lookup treats all orders of a fixed multiset as equivalent, whereas
insertion may interpret those orders as control states.

\subsection{Residual Paths and Two Kinds of Failure}

Fix a valid table state. Its directed residual multigraph has one vertex per
bucket. Each fingerprint $f$ stored in bucket $v$ contributes a movable edge
$v\rightarrow\operatorname{alt}(v,f)$; parallel edges are retained. A bucket
with fewer than four residents is a \emph{cavity}. If both roots of a new
fingerprint are full, insertion succeeds by finding a directed path from either
root to a cavity and shifting residents along that path.

\begin{lemma}[Residual-path equivalence]
\label{lem:residual-equivalence}
Relative to a valid placement, a new item is insertable by reassignment if and
only if one of its candidate roots has a directed residual path to a cavity.
\end{lemma}
The path may have length zero, so a non-full root is an immediate witness.
\begin{proof}
A residual path moves each resident to its alternate endpoint and terminates in
a free slot, producing an augmenting reassignment. Conversely, expand every
bucket into four slot vertices and consider an augmenting path in the
item--slot matching. After the new item chooses a root, each alternating matched
edge identifies a resident in the current bucket and the following unmatched
edge moves it to that resident's other endpoint. Collapsing the four slots of
each bucket yields a residual path to the terminal cavity.
\end{proof}

This equivalence separates \emph{structural capacity} from \emph{algorithmic
capacity}. A prefix is structurally feasible when all items admit a capacity-four
orientation over their endpoint pairs. A bounded online policy may fail earlier
because it does not find an augmenting path within its relocation budget. The
exact oracle in Section~\ref{sec:oracle} measures the capacity still present
when online search stops.

\section{CavityRank}
\label{sec:cavityrank}

\method{} adds direction to a single relocation path in three steps. It encodes
a four-level rank in each full bucket, routes toward the minimum-ranked
alternate, and backs up every modified bucket from its realized post-state. The
ordinary choice between non-full roots remains unchanged; when a direct insert
first fills a bucket, that bucket performs the same post-state backup.

\subsection{A Four-State Implicit Rank}
\label{sec:codec}

Let full bucket $v$ contain fingerprints $a_0,\ldots,a_3$. \method{} decodes
\begin{equation}
  R_4(v)=1+\indicator{a_1<a_0}+2\cdot\indicator{a_3<a_2}.
  \label{eq:rank-codec}
\end{equation}
The orientations of lane pairs $(0,1)$ and $(2,3)$ provide two comparison bits.
The encoder chooses a deterministic partition into two unequal pairs whenever
one exists and orients each pair to realize the requested bits. Every operation
permutes lanes only inside the bucket, so
Theorem~\ref{thm:query-invariance} preserves lookup and the one-word footprint.

Figure~\ref{fig:codec} shows the ordinary code for four distinct fingerprints.
The construction uses four of the 24 lane permutations, trading codebook size
for a constant-time encoder with a simple duplicate boundary. When no
fingerprint occurs more than twice, the four values admit two unequal pairs and
all four states are realizable. Three or four equal values force at least one
tied comparison; the implementation then uses a fixed sorted fallback.

\begin{figure}[t]
  \centering
  \small
  \setlength{\fboxsep}{2.7pt}
  \renewcommand{\arraystretch}{1.28}
  \begin{tabular*}{\columnwidth}{@{\extracolsep{\fill}}c c c@{}}
    \toprule
    $R_4(v)$ & lane pairs $(0,1)\mid(2,3)$ & comparison bits \\
    \midrule
    1 & \bucketfour{a}{b}{c}{d} & $(0,0)$ \\
    2 & \bucketfour{b}{a}{c}{d} & $(1,0)$ \\
    3 & \bucketfour{a}{b}{d}{c} & $(0,1)$ \\
    4 & \bucketfour{b}{a}{d}{c} & $(1,1)$ \\
    \bottomrule
  \end{tabular*}
  \vspace{3pt}

  \parbox{\columnwidth}{\centering\footnotesize
  $a<b<c<d$; every row stores the same multiset $\mset{a,b,c,d}$.
  For $\mset{a,a,a,b}$, the fixed sorted fallback preserves the multiset and
  decodes as state 1.}
  \caption{Two pair orientations encode a four-level rank without changing the
  resident multiset.}
  \label{fig:codec}
\end{figure}

Write $E_4(B,r)$ for encoding requested state $r$ from packed ordering $B$, and
$D_4(B)$ for the decoder in Equation~\ref{eq:rank-codec}. The one-bit ablation
\crtwo{} uses $D_2(B)=1+\indicator{a_1<a_0}$ and the corresponding encoder
$E_2$. For $q\in\{2,4\}$, write $R_q(v)=D_q(B_v)$ for the rank realized by a
full bucket's current packed ordering.

\begin{lemma}[Codec realizability]
\label{lem:codec-realizability}
$E_4$ preserves the fingerprint multiset. If no fingerprint has multiplicity
above two, $D_4(E_4(B,r))=r$ for every $r\in\{1,2,3,4\}$. With multiplicity
three or four, the fixed sorted fallback preserves the multiset and realizes
state 1.
\end{lemma}
\begin{proof}
With maximum multiplicity two, the four values can be partitioned into two
unequal pairs, and each pair can be oriented independently. Three equal values
force one comparison to tie, so the implementation uses the deterministic
multiset-preserving fallback. Safety follows from
Theorem~\ref{thm:query-invariance}.
\end{proof}

\subsection{Minimum-Rank Single-Path Routing}
\label{sec:routing}

For codec size $q\in\{2,4\}$ and candidate target bucket $u$, define
\begin{equation}
  s_q(u)=
  \begin{cases}
    0, & \text{if $u$ is non-full},\\
    R_q(u), & \text{if $u$ is full}.
  \end{cases}
  \label{eq:score}
\end{equation}
Score 0 denotes a cavity reachable in the next move. The direct fast path uses
a non-full root immediately: if both roots have space, it chooses the lower
occupancy and uses a fixed hash bit on an occupancy tie. If both roots are full,
\method{} starts from the lower-ranked root and uses a second fixed hash bit on
an equal-rank tie. Neither hash tie-break inspects a lane index.

At full bucket $v$ and relocation step $t$, let $f_{\mathrm{carry}}$ be the carried
fingerprint, and let resident $f_j$ point to
$u_j=\operatorname{alt}(v,f_j)$. The algorithm evicts a resident whose target
has minimum score. Equal scores are resolved by lexicographically minimizing
\begin{equation}
  \bigl(s_q(u_j),H(f_{\mathrm{carry}},v,t,u_j,f_j),u_j,f_j\bigr),
  \label{eq:canonical-tie}
\end{equation}
where $H$ is a fixed deterministic 64-bit function of the listed semantic
arguments. The tie rule contains no lane index, so changing the codec alphabet
does not alter equal-score actions through lane ordering. A remaining tie
represents the same placement action.

The rank is ordinal: a state-1 target is preferred to state 2, but the value
does not estimate an exact path length. This limited semantics permits a
two-bit, path-local representation.

\subsection{Backing Up the Graph After Relocation}
\label{sec:post-update}

The backup must describe the graph after the swap, not the graph inspected
before it. Suppose the carried fingerprint arrives at full bucket $v$ from
predecessor $p$. The old residents point to $u_0,\ldots,u_3$, and resident
$j^*$ is evicted. The swap removes edge $v\rightarrow u_{j^*}$ and installs the
incoming fingerprint, whose alternate from $v$ is $p$.

\begin{lemma}[Post-relocation residual multiset]
\label{lem:post-state}
After the swap, the outgoing target multiset of $v$ is
\begin{equation}
  \mathcal N_{\mathrm{after}}(v)
  =\mset{p}\uplus\mset{u_j\mid j\neq j^*}.
  \label{eq:post-neighbors}
\end{equation}
\end{lemma}
\begin{proof}
The evicted fingerprint deletes its edge to $u_{j^*}$. By
Equation~\ref{eq:alternate-involution}, the incoming fingerprint contributes
an edge back to the predecessor $p$. The other three resident edges remain.
\end{proof}

\begin{figure}[t]
\centering
\begin{tikzpicture}[
  font=\scriptsize,
  node/.style={circle,draw,minimum size=5.5mm,inner sep=0pt},
  edge/.style={-{Latex[length=1.7mm]},gray!65,thin},
  victim/.style={-{Latex[length=1.9mm]},deepred,very thick},
  incoming/.style={-{Latex[length=1.9mm]},deepblue,very thick}
]
\begin{scope}[xshift=0cm]
  \node[font=\bfseries] at (1.55,2.03) {before};
  \node[node] (p1) at (0.02,1.00) {$p$};
  \node[node] (v1) at (1.38,1.00) {$v$};
  \node[node] (u1) at (3.00,1.90) {$u_{j^*}$};
  \node[node] (a1) at (3.00,1.30) {$u_a$};
  \node[node] (b1) at (3.00,0.70) {$u_b$};
  \node[node] (c1) at (3.00,0.10) {$u_c$};
  \draw[incoming] (p1)--node[above] {$f_{\mathrm{in}}$} (v1);
  \draw[victim] (v1)--node[above,sloped] {$f_{j^*}$} (u1);
  \draw[edge] (v1)--(a1); \draw[edge] (v1)--(b1); \draw[edge] (v1)--(c1);
\end{scope}
\begin{scope}[xshift=4.20cm]
  \node[font=\bfseries] at (1.55,2.03) {after};
  \node[node] (p2) at (0.02,1.00) {$p$};
  \node[node] (v2) at (1.38,1.00) {$v$};
  \node[node] (a2) at (3.00,1.60) {$u_a$};
  \node[node] (b2) at (3.00,1.00) {$u_b$};
  \node[node] (c2) at (3.00,0.40) {$u_c$};
  \draw[incoming] (v2)--node[above] {$f_{\mathrm{in}}$} (p2);
  \draw[edge] (v2)--(a2); \draw[edge] (v2)--(b2); \draw[edge] (v2)--(c2);
\end{scope}
\end{tikzpicture}
\caption{Post-state backup. The victim edge disappears, the incoming
fingerprint adds the predecessor edge, and all three unaffected edges remain.}
\label{fig:post-state}
\end{figure}
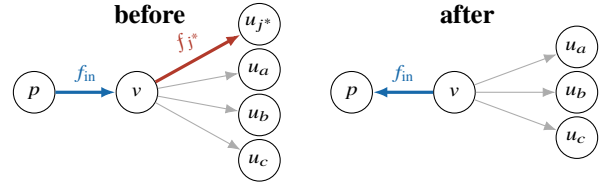

Set $q=4$ for \method{} and $q=2$ for \crtwo{}. Write $B_v^+$ for the packed
ordering immediately after a resident change. For a relocation,
$\mathcal N_{\mathrm{after}}(v)$ is given by
Equation~\ref{eq:post-neighbors}. For a direct insertion that first fills a
bucket, it is the target multiset of all four post-insertion residents. Both
cases use the same update:
\begin{align}
  \widehat R_q(v)
  &=\min\!\left(q,1+\min_{u\in\mathcal N_{\mathrm{after}}(v)}s_q(u)\right),
  \label{eq:backup-request}\\
  B_v&\leftarrow E_q(B_v^+,\widehat R_q(v)),
  \qquad R_q(v)=D_q(B_v).
  \label{eq:backup-realized}
\end{align}
Here $\widehat R_q(v)$ is the requested rank and $R_q(v)$ is the realized rank.
The next relocation propagates $R_q(v)$, preserving the algorithm's actual
state when duplicate fingerprints make requested states alias. Using all four
old targets would retain the deleted victim edge and omit the newly created
predecessor edge.

\subsection{Algorithmic Contract}
\label{sec:contract}

\method{} maintains a bounded asynchronous ordering signal. Ranks saturate at
four, and each insertion updates only buckets on its path. The signal ranks
local choices rather than maintaining an exact distance. Its persistent state
is entirely the packed lane order; an insertion uses only constant transient
state for the carry, predecessor, relocation step, and tie key. The policy
therefore requires no frontier, table-wide repair, per-bucket label array, or
other table-scaled workspace.

The relocation chain maintains a simple invariant. Before each step, only the
carried fingerprint is unplaced; every bucket respects capacity; and the current
bucket is a legal endpoint of the carry. Swapping preserves the invariant with
the victim as the new carry, while appending at a cavity restores a complete
placement. If the relocation budget is exhausted, the incomplete table state is
discarded rather than reused.

\section{Exact Evaluation Boundary}
\label{sec:method}

\subsection{A Complete Capacity-Four Oracle}
\label{sec:oracle}

Let $V=\{0,\ldots,m-1\}$ be the buckets and let $\mathcal I$ index an item
prefix. Item $e\in\mathcal I$ has endpoint pair
$\Gamma(e)=(u_e,v_e)\in V^2$, allowing $u_e=v_e$. A feasible orientation
$\phi:\mathcal I\rightarrow V$ assigns every item to one endpoint and respects
capacity:
\begin{equation}
  \begin{aligned}
    \phi(e)&\in\{u_e,v_e\} &&(e\in\mathcal I),\\
    \bigl|\{e\in\mathcal I:\phi(e)=v\}\bigr|&\leq 4 &&(v\in V).
  \end{aligned}
  \label{eq:orientation}
\end{equation}
The incremental oracle places a new item in a direct vacancy when available;
otherwise it follows currently assigned items to a non-full bucket and reverses
the resulting augmenting path.

\begin{theorem}[Oracle completeness]
\label{thm:oracle}
Given a feasible orientation of a prefix, the oracle accepts the next item if
and only if the extended prefix has a capacity-four orientation.
\end{theorem}
\begin{proof}
Expand each bucket $v$ into four slot vertices and connect every item to all
four slots of each endpoint. Equation~\ref{eq:orientation} is equivalent to a
matching that covers all items. Adding one item preserves the matching for the
old prefix. By Berge's theorem, a covering extension exists exactly when an
augmenting path from the new item reaches an unmatched slot. Grouping slots by
bucket gives the oracle's BFS to a non-full bucket, and reversing that path
updates the orientation. Oracle failure therefore proves structural
infeasibility.
\end{proof}

Let experimental cell $C=(G,m,k)$ fix the graph model, number of buckets, and
relocation budget. On paired stream $\ell$, let $N_{\mathrm{oracle},\ell}$ and
$N_{M,\ell}$ be the accepted-prefix lengths of the oracle and online method
$M$, and define the item gap
$g_{M,\ell}=N_{\mathrm{oracle},\ell}-N_{M,\ell}$. Within one cell, item-gap and
load-gap closure are identical because every load divides by the common factor
$4m$. Relative to baseline $M_0$, the population estimand is
\begin{equation}
  \theta_C(M;M_0)
  =1-\frac{\mathbb E_C[g_M]}{\mathbb E_C[g_{M_0}]},
  \qquad \mathbb E_C[g_{M_0}]>0,
  \label{eq:closure}
\end{equation}
and the paired estimator is the ratio of aggregate gaps,
\begin{equation}
  \widehat\theta_C(M;M_0)
  =1-\frac{\sum_{\ell=1}^{n_C}g_{M,\ell}}
           {\sum_{\ell=1}^{n_C}g_{M_0,\ell}},
  \qquad \sum_{\ell=1}^{n_C}g_{M_0,\ell}>0.
  \label{eq:closure-estimator}
\end{equation}
If the denominator is zero, the closure is reported as not available;
otherwise, the estimator weights streams by the baseline gap rather than
averaging per-stream ratios. Paired bootstrap resampling draws streams jointly
across methods and recomputes Equation~\ref{eq:closure-estimator}.

\subsection{Policies and Study Design}
\label{sec:design}

Table~\ref{tab:policies} separates persistent routing state, dedicated
insertion workspace, and search pattern while holding the 8-byte payload bucket
and two-bucket query fixed. \crtwo{} is the controlled one-bit ablation;
\method{} is the two-bit policy. CavityScan isolates the one-step look-ahead
used by Vacuum Filters: at each full bucket it scans all resident alternates and
takes a direct cavity when one exists \citep{wang2019vacuum}. Local Search
Allocation (LSA) supplies an explicit unsaturated label, while depth-10 BFS
supplies a frontier and parent state.

\begin{table*}[t]
\centering
\footnotesize
\setlength{\tabcolsep}{4.2pt}
\renewcommand{\arraystretch}{1.12}
\begin{tabularx}{0.97\textwidth}{@{}l >{\raggedright\arraybackslash}X c c >{\raggedright\arraybackslash}X@{}}
\toprule
Policy & Routing signal & Extra persistent bytes & Dedicated workspace & Search pattern \\
\midrule
Random CF / Rotor & none / bucket-local rotation & 0 B/bucket & 0 B & one relocation path \\
CavityScan & direct cavity among alternates & 0 B/bucket & 0 B & one-hop lookahead \\
\crtwo{} & one implicit comparison bit & 0 B/bucket & 0 B & two-level ordinal path \\
\method{} & two implicit comparison bits & 0 B/bucket & 0 B & four-level ordinal path \\
LSA & explicit unsaturated label & 8 B/bucket & 0 B & minimum-label path \\
BFS & frontier and parent state & 0 B/bucket & $10m+\mathcal O(1)$ B & depth-10 frontier search \\
\bottomrule
\end{tabularx}
\caption{Representation and search contracts. Dedicated workspace excludes
constant registers and stack variables used by every insertion routine.}
\label{tab:policies}
\end{table*}

A preregistered simulator ladder places Random CF, \crtwo{}, \method{}, LSA,
depth-10 BFS, and the exact oracle on the same 2,048 streams at 4,096 buckets
and a 5,000-relocation budget. XOR16 matches the packed alternate rule; an
independent-edge model supplies a stress replicate. This ladder uses the
earlier lane-coupled equal-score rule and is therefore descriptive.

A separate canonical-tie campaign isolates the state alphabet. It analyzes
XOR16 and Keyed-XOR16, whose offset is separately keyed, at 4,096 buckets and
2,048 paired streams, then at 65,536 buckets and 256 paired streams. Within
each model, \crtwo{} and \method{} receive the same streams, use the semantic
tie rule in Equation~\ref{eq:canonical-tie}, and propagate realized ranks.
Their only algorithmic difference is the one-bit versus two-bit codec; both
cells use a 5,000-relocation budget.

The packed Rust campaign measures correctness, logical bucket reads and writes,
throughput, and insertion latency on an Apple M4 Pro with 48~GB memory. Its
64~MiB build cell targets 97.75\% load. This implementation uses the earlier
lane-coupled equal-score rule, so the campaign supplies machine-cost and
correctness evidence rather than the isolated state-alphabet estimate.
Appendix~\ref{app:protocol} records the seed spaces, resampling rules, machine
configuration, and oracle cross-checks.

\section{Evaluation}
\label{sec:evaluation}

\subsection{Recovering the Search Gap}
\label{sec:eval-gap}

\begin{figure*}[!t]
\centering
\includegraphics[width=0.96\textwidth]{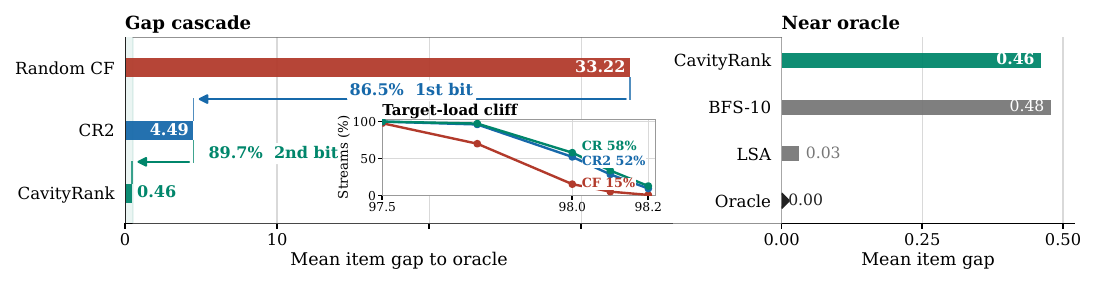}
\caption{XOR16 gap cascade at 4,096 buckets, 5,000 relocations, and 2,048
paired streams. Bars share an oracle origin and their lengths equal mean item
gaps; arrows report closure. Right: same-unit expansion of the boxed 0--0.52
interval (resource reference, not ordering). Inset: same-stream target reach,
labeled at the 98\% cliff.}
\label{fig:search-gap}
\end{figure*}

Figure~\ref{fig:search-gap} restores the missing denominator. On XOR16, Random
CF stops a mean 33.22 items before the exact oracle. \crtwo{} reduces that gap
to 4.49 items, closing 86.47\%; \method{} reduces it again to 0.46 items,
closing 89.73\% of \crtwo{}'s residual and 98.61\% of the original gap. The
independent-edge model repeats the decomposition: $32.70\rightarrow
4.41\rightarrow0.43$ items, or 86.51\% followed by 90.23\% residual closure.
The gains compound: after the first bit removes 86.5\% of the gap, the 98\%
reach rate rises from 15.33\% under Random CF to 52.20\% under \crtwo{} and
57.91\% under \method{}.

The canonical-tie ablation reduces the XOR16 aggregate gap from
\CausalMainXorPairGapSum{} to \CausalMainXorCRFourGapSum{} items, closing
\CausalMainXorClosurePct{} (95\% CI
[\CausalMainXorClosureCILowPct{}, \CausalMainXorClosureCIHighPct{}]); \method{}
wins/loses/ties on \CausalMainXorWins{}/\CausalMainXorLosses{}/
\CausalMainXorTies{} streams. Keyed-XOR16 independently closes
\CausalMainKeyedClosurePct{} (95\% CI
[\CausalMainKeyedClosureCILowPct{}, \CausalMainKeyedClosureCIHighPct{}]), with
\CausalMainKeyedWins{}/\CausalMainKeyedLosses{}/\CausalMainKeyedTies{}. This
replay attributes the ladder's CR2-to-\method{} step to the second bit rather
than the lane-coupled tie regime.

At the fixed 98\% target, the reach probability increases by
\CausalMainXorReachDiffPP{} percentage points (95\% CI
[\CausalMainXorReachDiffCILowPP{}, \CausalMainXorReachDiffCIHighPP{}]) on
XOR16 and by \CausalMainKeyedReachDiffPP{} percentage points (95\% CI
[\CausalMainKeyedReachDiffCILowPP{}, \CausalMainKeyedReachDiffCIHighPP{}]) on
Keyed-XOR16. The mean gain is
\CausalMainXorGapMean{} and \CausalMainKeyedGapMean{} accepted items per stream,
but those items arrive at a sharp near-threshold boundary; a small load shift
therefore changes the fraction of streams that cross the target.

The effect remains large at 65,536 buckets. \method{} closes
\CausalScaleXorClosurePct{} (95\% CI
[\CausalScaleXorClosureCILowPct{}, \CausalScaleXorClosureCIHighPct{}]) of the
XOR16 gap and \CausalScaleKeyedClosurePct{} (95\% CI
[\CausalScaleKeyedClosureCILowPct{}, \CausalScaleKeyedClosureCIHighPct{}]) of the
Keyed-XOR16 gap, winning all 256 paired streams in both models. The
corresponding increases in 98\%-load reach probability are
\CausalScaleXorReachDiffPP{} and \CausalScaleKeyedReachDiffPP{} percentage
points. The mean item gains, \CausalScaleXorGapMean{} and
\CausalScaleKeyedGapMean{}, equal about
0.041--0.042 percentage points of load in a 262,144-slot table, yet they move
many streams across the fixed threshold. This threshold amplification explains why a
modest capacity shift can produce a large reach-probability difference.

The tie-regime sensitivity check changes closure by
$+\CausalMainXorInteractionPP$ percentage points (90\% CI
$[\CausalMainXorInteractionLowPP,\CausalMainXorInteractionHighPP]$) on XOR16
and $+\CausalMainKeyedInteractionPP$ percentage points (90\% CI
$[\CausalMainKeyedInteractionLowPP,\CausalMainKeyedInteractionHighPP]$) on
Keyed-XOR16. Both intervals lie inside the predefined
$\pm 5$-percentage-point margin.
Within the headline cell, the second codec bit rather than lane-coupled tie
behavior drives the measured closure. Appendix~\ref{app:contrast} gives the
complete protocol and sensitivity figure.

\subsection{Packed Cost and Memory Traffic}
\label{sec:eval-cost}

All packed policies in Figure~\ref{fig:machine} reach 97.75\% load in all 16
runs of the 64~MiB cell. \method{} uses 42.53 logical reads and 10.60 writes per
insertion. LSA uses 62.67 reads and 10.61 writes while allocating a 64~MiB label
array. Depth-10 BFS uses 355.74 reads and 3.78 writes while reserving 80.1~MiB of
observed frontier and parent workspace. Relative to these guided baselines,
\method{} reduces read traffic by 32\% and 88\%, respectively, and removes all
table-scaled auxiliary allocation. BFS's shorter relocation paths write fewer
buckets; \method{} exchanges more writes for far fewer frontier reads.

\begin{figure*}[t]
\centering
\includegraphics[width=0.95\textwidth]{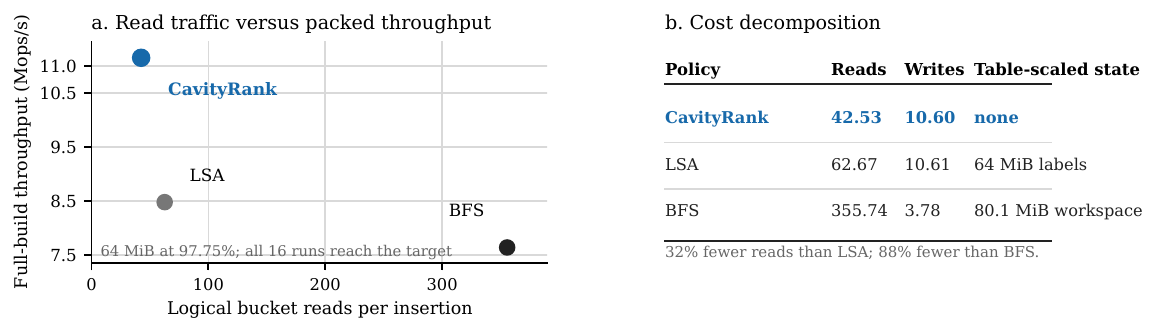}
\caption{Packed machine-cost comparison at 64~MiB and 97.75\% load. Left:
logical read traffic and full-build throughput. Right: reads, writes, and
table-scaled state. Every policy reaches the target in all 16 runs.}
\label{fig:machine}
\end{figure*}

The clean-checkout artifact reaches 90\% load with zero false negatives in all
\ResultCorrectnessRows{} small-table build rows. Across 32 rows of one million queries each, \method{} and its optional guarded
variant again report zero false negatives; within each of
\ResultQueryEquivalentPairs{} paired runs, they produce identical false-positive
counts. Query transparency is structural by
Theorem~\ref{thm:query-invariance}; these checks verify the packed
implementation. Together, the capacity and machine studies place \method{} at
the intended design point: persistent residual direction without a frontier,
per-bucket label array, or wider query representation.

\section{Related Work}
\label{sec:related}

\paragraph{Packed cuckoo filters and relocation.}
Cuckoo hashing establishes two-choice placement and relocation
\citep{pagh2004cuckoo}; bucketized variants pack several residents into each
candidate location \citep{dietzfelbinger2007balanced}. Cuckoo filters replace
full keys with fingerprints while preserving alternate recovery, two-bucket
lookup, and deletion \citep{fan2014cuckoo}. Overlapping-window filters change
the placement substrate to reduce space overhead \citep{schmitz2026smaller}.
\method{} keeps the standard four-slot partial-key substrate and instead uses
query-equivalent order to guide insertion.

\paragraph{Single-path guidance.}
Practical policies attach local history or heuristics to a cheap relocation
path. Better Choice balances the initial buckets \citep{wang2019betterchoice};
queue kicking uses bucket-local age \citep{kuszmaul2016kickout}; MinCounter and
CostCounter keep explicit historical counters
\citep{sun2015mincounter,wu2023costcounter}. Necklace and SmartCuckoo maintain
auxiliary path or graph information to avoid repeated kick-outs
\citep{li2014necklace,sun2017smartcuckoo}, while Local Minimum uses an
eviction-count signal \citep{zhang2025localminimum}. \method{} instead
reads the current ranks of alternate buckets and stores the resulting local
ordering in query-equivalent lane order.

\paragraph{Local labels and explicit search.}
Local Search Allocation uses explicit integer labels to guide insertion without
running BFS at every step \citep{khosla2019faster}; related work adapts such
labels to overlapping blocks \citep{walzer2023overlapping}. Minimum-label
routing is therefore prior art. \method{} studies the finite-state packed point:
four ordinal states in one four-fingerprint bucket, updated only along the
online path. Breadth-first relocation remains the stronger search primitive
when frontier cost is acceptable \citep{li2014algorithmic}; Cuckoo-GPU adapts
BFS to GPU execution \citep{dortmann2026cuckoogpu}.

\paragraph{Theory and implicit state.}
Random-walk analyses distinguish peeling and orientation thresholds and prove
strong insertion bounds in $d$-ary models
\citep{walzer2022randomwalk,bellfrieze2024randomwalk,bellfrieze2025randomwalkrevised}.
Bubbling Up uses several
full-key candidate locations at high load \citep{kuszmaul2025bubbling}, and Bell
and Kuszmaul preserve unrevealed second-hash randomness for sufficiently large
buckets \citep{bellkuszmaul2026bucketized}. Those full-key freshness mechanisms
and \method{}'s exposed partial-key residual signal use different resources.
More broadly, data arrangement can carry computational state
\citep{bender2025nonoblivious}. \method{} applies that principle locally: the
permutation class of one four-slot bucket supplies a small dynamic control
alphabet.

\section{Discussion}
\label{sec:discussion}

\paragraph{Equivalence is the resource.}
The 64-bit bucket carries rank through multiple representations that lookup
already treats as identical. This principle extends beyond cuckoo filters:
whenever a query observes an equivalence class rather than a unique
representation, unused degrees of freedom can carry local control state. The
update rule must follow the post-operation semantics of the data structure; the
predecessor edge in Figure~\ref{fig:post-state} is the decisive example.

\paragraph{Four ranks recover most search loss.}
The full ladder shows the progression: one bit removes 86.5\% of Random CF's
gap, and four ranks leave only 1.3--1.4\% of that original gap. The controlled
ablation attributes roughly 90\% of \crtwo{}'s residual to the second bit at
4,096 buckets and 84\% at 65,536 buckets. The improvement does not require an
exact distance label. Under the same 5,000-relocation budget, the larger table
leaves a broader residual tail, but the compact signal remains effective across
a sixteen-fold scale increase.

\paragraph{Representation and lifecycle scope.}
The current codec targets four-slot partial-key buckets. Fingerprint order is
also used by semi-sorting in the original Cuckoo filter, so the two encodings do
not compose without redesigning the representation. Deletion preserves
membership but may leave nearby ranks stale because a newly opened cavity does
not trigger backward repair. The present implementation is single-threaded and
insert-mostly; concurrent relocation would require synchronization around
packed-bucket updates and relocation chains. These boundaries point to
representation-aware composition and lifecycle maintenance as the next design
questions.

The optional insertion-local novelty guard is analyzed in
Appendix~\ref{app:guard}. It recovers a small, cell-dependent part of the
remaining tail and is unresolved at the largest tested scale, so the packed
four-state core remains the default algorithm.

\section{Conclusion}

Near-threshold cuckoo insertion need not choose only between an uninformed kick
and a full frontier. \method{} turns query-equivalent fingerprint order into a
four-level residual rank, follows one minimum-rank path, and updates each bucket
from the graph that exists after relocation. The second implicit bit closes
about 90\% of \crtwo{}'s search gap at 4,096 buckets, leaving 1.39\% of Random
CF's original XOR16 gap, and closes about 84\% at 65,536. The packed
implementation uses 42.53 reads per insertion, no extra
per-bucket bytes, and no table-scaled workspace. The broader lesson is that
representation equivalence is an algorithmic resource: a data structure can
carry control information without changing what its queries observe.

\begingroup
\footnotesize
\sloppy
\bibliographystyle{plainnat}
\bibliography{references}
\endgroup

\appendix

\section{Controlled Contrast and Sensitivity}
\label{app:contrast}

The controlled campaign was fixed before observation and uses seeds disjoint
from the implementation-cost and tail-diagnostic campaign. The 4,096-bucket
cell uses seeds 80,000--82,047; the 65,536-bucket scale cell uses seeds
90,000--90,255. For each 4,096-bucket seed and model, \crtwo{} and \method{}
run under both lane-coupled and permutation-invariant tie regimes. All arms
propagate realized decoded ranks.
Within a tie regime, the two policies differ only in their one-bit versus
two-bit codec.

XOR16 is the primary partial-key model and Keyed-XOR16 is a separate
confirmation; all estimates remain model-specific. The analysis uses 20,000
clustered bootstrap draws over absolute seeds and recomputes both closure and
tie-regime interaction. Two-sided 95\% intervals accompany closure estimates. The predefined
sensitivity analysis uses a joint-bootstrap 90\% interval and a
$\pm 5$-percentage-point margin. The
98\%-load reach difference is paired and uses a joint-bootstrap 95\% interval.
Pilot seeds 70,000--70,255 do not enter inference.

\begin{figure}[H]
\centering
\includegraphics[width=\columnwidth]{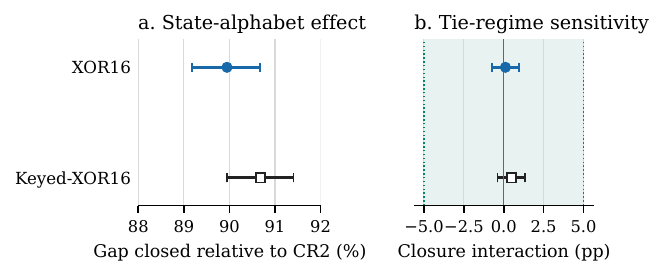}
\caption{State-alphabet effect and tie-regime sensitivity at 4,096 buckets. The
interaction intervals lie inside the predefined $\pm 5$-percentage-point
margin.}
\label{fig:lane-sensitivity}
\end{figure}

\begin{table}[H]
\centering
\small
\setlength{\tabcolsep}{3.1pt}
\renewcommand{\arraystretch}{1.15}
\begin{tabular}{@{}lcc@{}}
\toprule
Metric & XOR16 & Keyed-XOR16 \\
\midrule
Gap closed & \CausalScaleXorClosurePct{} & \CausalScaleKeyedClosurePct{} \\
95\% CI & [\CausalScaleXorClosureCILowPct{}, \CausalScaleXorClosureCIHighPct{}]
& [\CausalScaleKeyedClosureCILowPct{}, \CausalScaleKeyedClosureCIHighPct{}] \\
Mean gain (items/stream) & \CausalScaleXorGapMean{} & \CausalScaleKeyedGapMean{} \\
Wins/losses/ties & 256/0/0 & 256/0/0 \\
98\%-reach difference (pp) & \CausalScaleXorReachDiffPP{} & \CausalScaleKeyedReachDiffPP{} \\
95\% CI (pp) & [\CausalScaleXorReachDiffCILowPP{}, \CausalScaleXorReachDiffCIHighPP{}]
& [\CausalScaleKeyedReachDiffCILowPP{}, \CausalScaleKeyedReachDiffCIHighPP{}] \\
\bottomrule
\end{tabular}
\caption{Scale confirmation over 256 paired streams per model at 65,536 buckets
and 5,000 relocations. Reach entries are absolute probability differences,
reported in percentage points.}
\label{tab:scale}
\end{table}

Expanded-slot matching cross-checks cover 1,024 deterministic package-level
instances and a disjoint 512-instance controlled preflight, including duplicate
fingerprints, self-loops, duplicate items, and parallel self-loops. These tests
validate the implementation; Theorem~\ref{thm:oracle} supplies completeness.

\section{Optional Tail Guard}
\label{app:guard}

Four-state ranks saturate, leaving tied minimum-score edges. An exact diagnostic
tracks the complete insertion history and prefers unseen tied targets. The packed
\guard{} variant approximates this signal with a 2~KiB Bloom sketch per insertion
thread. It activates after 128 relocations, records only the post-activation
suffix, clears between insertions, and never enters lookup. False positives may
hide an unseen target, but they cannot change the stored multiset or placement
validity.

At 4,096 buckets and 20,000 relocations, the guard closes
\ResultCapHTwoIndependentGapClosurePct{} of the core's residual gap in the
explicit-edge stress model and \ResultCapHTwoXorSixteenGapClosurePct{} in
XOR16. The absolute gains are \ResultCapHTwoIndependentGapItemsMean{} and
\ResultCapHTwoXorSixteenGapItemsMean{} items per stream, with median zero and
more than 1,800 ties in each 2,048-pair comparison. At 65,536 buckets the effect
falls to single-digit or low-double-digit closure depending on budget and
model; at 1,048,576 buckets both intervals cross zero and both policies reach
98\% on all 32 streams.

\begin{figure}[H]
\centering
\includegraphics[width=\columnwidth]{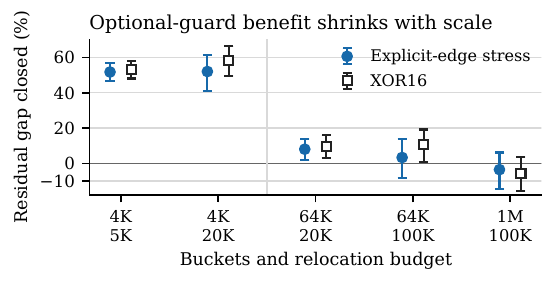}
\caption{The optional guard's residual-gap benefit is cell-dependent, shrinks
with scale, and is unresolved at 1,048,576 buckets.}
\label{fig:path-scale}
\end{figure}

The guard isolates one source of rare residual failures, but its common-case
throughput and latency ratios remain near one and the core four-state policy
remains the default.

\section{Implementation Protocol}
\label{app:protocol}

The packed implementation campaign uses the earlier lane-coupled policy package
rather than the semantic tie rule in Equation~\ref{eq:canonical-tie}. It
measures machine traffic, memory, latency, and correctness. The controlled
simulator campaign uses a separate seed space to estimate the state-alphabet
effect; the two campaigns answer distinct questions and are analyzed separately.

The machine study runs on an Apple M4 Pro with 48~GB memory, Rust 1.92.0, thin
LTO, and one codegen unit. A logical read is one explicit access to a
packed bucket word; a logical write is one replacement of that word. These
counts expose algorithmic memory traffic rather than hardware cache misses.
The 64~MiB, 97.75\% build cell uses 16 independent seeded runs. The packed
artifact allocates zero extra persistent bytes per bucket for \method{},
64~MiB of labels for LSA, and 80.1~MiB of observed workspace for BFS.

The tail-diagnostic campaign uses pilot seeds 0--127 to select the 2~KiB guard
with activation after 128 relocations and a full reset between insertions.
Pilot rows never enter final intervals. Final capacity cells use seeds
10,000--12,047 with 5,000 and 20,000 relocations at 4,096 buckets;
20,000--20,255 with 20,000 and 100,000 relocations at 65,536 buckets; and
30,000--30,031 with 100,000 relocations at 1,048,576 buckets. Separate
diagnostic cells use seeds
40,000--40,255 and a 99\% load cap; they do not contribute capacity replicates.
Depth-10 BFS runs once per table size, model, and seed because it has no online
relocation budget; analysis maps that row to the corresponding budget cells and
rejects duplicates. Any unexpectedly censored exact block is marked invalid
rather than replaced with a new seed. All final configurations, raw schemas,
ordering rules, compiler records, checksums, and headline macros are fixed in the artifact.

\end{document}